\documentclass{elsarticle}

\usepackage{amsmath}
\usepackage{amssymb}
\usepackage{amsthm}
\usepackage{multirow}

\journal{Finite Fields and Their Applications}

\newtheorem{theorem}{Theorem}

\newtheorem{lemma}[theorem]{Lemma}

\newtheorem{example}[theorem]{Example}

\newtheorem{remark}[theorem]{Remark}

\newcommand{\bb}{\mathbb}
\newcommand{\Z}{\bb{Z}}

\newcommand{\cS}{\mathcal{S}}
\newcommand{\cC}{\mathcal{C}}
\newcommand{\cD}{\mathcal{D}}

\newcommand{\be}{\beta}
\newcommand{\de}{\delta}
\newcommand{\De}{\Delta}
\newcommand{\ga}{\gamma}

\newcommand{\zp}{\zeta_{p}}

\newcommand{\zs}{\{0\}}
\newcommand{\sm}{\setminus}

\newcommand{\supp}{\text{supp}}

\newcommand{\ul}{\underline}

\newcommand{\lan}{\langle}
\newcommand{\ran}{\rangle}

\newcommand{\wti}{\widetilde}

\newcommand{\Tr}{\text{Tr}}
\newcommand{\TQq}{\Tr^{Q}_{q}}

\newcommand{\Tqp}{\Tr^{q}_{p}}

\newcommand{\F}{\mathbb{F}}
\newcommand{\Fp}{\mathbb{F}_p}
\newcommand{\Fq}{\mathbb{F}_q}
\newcommand{\FQ}{\mathbb{F}_Q}

\newcommand{\ux}{\underline{x}}

\newcommand{\ub}{\underline{b}}
\newcommand{\uy}{\underline{y}}
\newcommand{\uu}{\underline{u}}
\newcommand{\subsp}{\left[{\FQ^t \atop r}\right]}
\newcommand{\dimq}{\dim_{\Fq}}
\newcommand{\Hp}{H^{\perp}}

\newcommand{\Hrp}{H_r^{\perp}}
\newcommand{\hm}{\frac{m}{2}}
\newcommand{\Fqhm}{\F_{q^{\hm}}}

\newcommand{\CN}{C^{(N,Q)}}

\newcommand{\etaN}{\eta^{(N,Q)}}

\newcommand{\etaNg}{\etaN_{g^h\sum_{j=1}^tb_j\be_j^h}}
\newcommand{\etaNy}{\etaN_{y_h}}

\begin{document}

\begin{frontmatter}

\title{The Weight Hierarchy of a Family of Cyclic Codes with Arbitrary Number of Nonzeroes}

\author[rvt]{Shuxing Li}
\ead{lsxlsxlsx1987@gmail.com}

\address[rvt]{Department of Mathematics, The Hong Kong University of Science and Technology,
Clear Water Bay, Kowloon, Hong Kong}

\begin{abstract}
The generalized Hamming weights (GHWs) are fundamental parameters of linear codes. GHWs are of great interest in many applications since they convey detailed information of linear codes. In this paper, we continue the work of \cite{XLG} to study the GHWs of a family of cyclic codes with arbitrary number of nonzeroes. The weight hierarchy is determined by employing a number-theoretic approach.
\end{abstract}


\end{frontmatter}

\section{Introduction}\label{sec1}

An $[n,k]$ linear code $\cC$ over finite field $\Fq$ is a $k$-dimensional subspace of the linear space $\Fq^n$. For any linear subcode $\cD \subset \cC$, the {\em support} of $\cD$ is defined to be
$$
\supp(\cD)=\{i: 0 \le i \le n-1, c_i \ne 0 \mbox{ for some } (c_0,c_1,\ldots,c_{n-1}) \in \cD \}.
$$
For $1 \le r \le k$, the $r$-th {\em generalized Hamming weight} (GHW) of $\cC$ is given by
$$
d_r(\cC)=\min \left\{|\supp(\cD)|: \cD \subset \cC \mbox{ and }  \dimq(\cD)=r\right\},
$$
where $|\supp(\cD)|$ denotes the cardinality of the set $\supp(\cD)$. By definition, $d_1(\cC)$ is just the minimum distance of $\cC$. The set $\{d_r(\cC): 1 \le r \le k\}$ is called the {\em weight hierarchy} of $\cC$.

The concept of GHWs was first introduced by Helleseth, Kl{\o}ve, Mykkeltveit \cite{HKM,K1} and was used in the computation of weight distributions. It was rediscovered by Wei \cite{Wei} to fully characterize the performance of linear codes when used in a wire-tap channel of type II or as a $t$-resilient function. Indeed, the GHWs provide detailed structural information of linear codes, which can also be used to compute the state and branch complexity profiles of linear codes \cite{F,KTFL}, to determine the erasure list-decodability of linear codes \cite{G} and so on.

In general, the determination of weight hierarchy is very difficult and there are only a few classes of linear codes whose weight hierarchies are known (see \cite{XLG} for a comprehensive enumeration of related references). This paper continues the work of \cite{XLG} to determine the weight hierarchy of a family of cyclic codes with arbitrary number of nonzeroes. Our result can be regarded as an extension of the results in \cite{HK2,V1,YLFL}, where the weight hierarchy of the semiprimitive codes was computed. We achieve this by generalizing a number-theoretic approach introduced in \cite{YLFL}.

The rest of this paper is organized as follows. In Section~\ref{sec2}, we introduce the concerned family of cyclic codes and state the main result. In Section~\ref{sec3}, we present a number-theoretic approach to the computation of GHWs. In Section~\ref{sec4}, we prove the main result. Section~\ref{sec5} concludes the paper.

\section{Main Result}\label{sec2}

In this section, we introduce the concerned family of cyclic codes and describe our main result.

At first, we set up some notations which will be used throughout the rest of the paper. Let $q=p^s$, $Q=q^m$, where $p$ is a prime, $s$ and $m$ are positive integers. Let $\ga$ be a primitive element of the finite field $\FQ$. We have the following three assumptions:
\begin{itemize}
\item[i)] $e \mid (Q-1)$, $a \not\equiv 0 \pmod{Q-1}$, $e \ge t \ge 1$;
\item[ii)] For $1 \le i \le t$, $a_i \equiv a+\frac{Q-1}{e}\De_i \pmod{Q-1}$. When $t\ge 2$, $\De_i \not\equiv \De_j \pmod{e}$ for $1 \le i,j \le t$, $i \ne j $ and $\gcd(\De_2-\De_1,\ldots,\De_t-\De_1,e)=1$;
\item[iii)] $\deg h_{a_i}(x)=m$ for $1 \le i \le t$ and $h_{a_i}(x) \ne h_{a_j}(x)$ for $i \ne j$. Here $h_a(x)$ denotes the minimal polynomial of $\ga^{-a}$ over $\Fq$.
\end{itemize}
Let us define
\begin{equation}\label{codepara}
\begin{split}
\de&=\gcd(Q-1,a_1,a_2,\ldots,a_t),\\
n&=\frac{Q-1}{\de}, \qquad N=\gcd\left(\frac{Q-1}{q-1},ae\right).
\end{split}
\end{equation}
When $e=t$, without loss of generality, we can choose $\De_i=i-1$ for $1 \le i \le t$.

We define $\cC$ to be the cyclic code of length $n$ over $\Fq$, whose parity-check polynomial is
$\prod_{i=1}^t h_{a_i}(x)$, where $a_i$'s are specified according to the three assumptions. Therefore, $\cC$ is an $[n,tm]$ cyclic code with $t$ nonzeroes.

This family of cyclic codes was first introduced in \cite{YXDL}, where the weight distributions were computed in several cases \cite{YXDL,YXX}. Due to the flexibility of the parameters $q$, $m$, $a$, $e$, $t$ and $\De_i$, this family contains an abundance of cyclic codes and some of which are interesting cyclic codes \cite{XLG}. In fact, \cite{YXDL} and \cite{YXX} presented a unified approach to the computation of weight distributions of certain cyclic codes, which included many previous results as special cases. Moreover, these results suggest that this family of codes is highly structured and it is hopeful to obtain more detailed information such as the generalized Hamming weights. Therefore, in \cite{XLG}, the authors obtained the weight hierarchy in the following cases:
\begin{itemize}
  \item $N=1,2$ and $e=t \ge 1$,
  \item $N=1$, $e>t \ge 1$ and $\{\De_1\pmod{e},\ldots,\De_t\pmod{e}\}$ is an arithmetic progression.
\end{itemize}
The computation relies heavily on generalizing a number-theoretic idea proposed in \cite{YLFL}. A key point in the computation is that, when $N=1$ or $2$, the evaluation of the corresponding Gauss periods is very simple. Note that the next simplest case for the evaluation of Gauss periods is the so-called semiprimitive case. Let $p$ be a prime and $N$ be an integer with $N>2$. In the semiprimitive case, there exists a positive integer $j$, such that $p^j \equiv -1 \pmod{N}$. In this paper, we consider the GHWs of $\cC$ in the semiprimitive case. More specifically, we have the following main result.

\begin{theorem}\label{thm-semi}
Let $p$ be a prime. Set $q=p^s$ and $Q=q^m$. Suppose $a$, $e$, $t$ and $a_i$, $1 \le i \le t$, are positive integers satisfying $e=t$ and the assumptions i) and ii). Let $n$ and $N$ be positive integers specified in \eqref{codepara}, satisfying $2 < N \le \sqrt{Q}$. Let $\cC$ be a cyclic code of length $n$, having parity-check polynomial $\prod_{i=1}^t h_{a_i}(x)$. Suppose $j$ is the smallest positive integer such that $p^j \equiv -1 \pmod{N}$. Suppose $\frac{sm}{2j}$ is odd and $m$ is even. For $1 \le r \le tm$, write $tm-r=r_1m+r_2$, where $0 \le r_1 \le t-1$ and $0 \le r_2 \le m-1$. Suppose $d_r:=d_r(\cC)$ is the $r$-th GHW of $\cC$. Then
\begin{equation*}
d_r=\begin{cases}
      (1-\frac{r_1}{t})\frac{q^m-1}{\de}-\frac{N(q^{r_2}-1)}{t\de} & \mbox{if $0 \le r_2 < \hm$,} \\
      (1-\frac{r_1}{t})\frac{q^m-1}{\de}-\frac{q^{r_2}-1+(N-1)(q^{\hm}-q^{r_2-\hm})}{t\de} & \mbox{if $\hm \le r_2 < m$.}
    \end{cases}
\end{equation*}
\end{theorem}

We have the following three remarks.

\begin{remark}
According to {\rm\cite[Lemma 6]{YXDL}}, the assumption iii) always holds true if $N \le \sqrt{Q}$. Hence, in Theorem~\ref{thm-semi}, we only need to choose integers $a$, $e$, $t$ and $a_i$, $1 \le i \le t$, that satisfying the assumptions i) and ii). Meanwhile, the condition $2<N \le \sqrt{Q}$ ensures that the assumption iii) also holds.
\end{remark}

\begin{remark}
In the semiprimitive case, by choosing $e=t=1$, $a=N$ and $\De_1=0$, the resulting code $\cC$ is simply a semiprimitive code. The GHWs of semiprimitive codes has been studied in {\rm\cite{HK2,V1,YLFL}}. More precisely, when the code $\cC$ is a semiprimitive code, Theorem~\ref{thm-semi} reduces to {\rm\cite[Theorem 3]{HK2}} and {\rm\cite[Theorem 4.1]{V1}} for $q=2$, and to {\rm\cite[Corollary 15]{YLFL}} for general $q$. When $t \ge 2$, i.e., the code $\cC$ is a reducible cyclic code, the result of Theorem~\ref{thm-semi} is new.
\end{remark}

\begin{remark}
The conditions $\frac{sm}{2j}$ being odd and $m$ being even are crucial. In fact, these two conditions are essentially used in the computation of the weight hierarchy of binary semiprimitive codes {\rm\cite{HK2,V1}}. Without them, the determination of the weight hierarchy, even for the simplest binary semiprimitive codes, remains a challenging problem.
\end{remark}

To confirm the correctness of Theorem \ref{thm-semi}, we provide some numerical examples, which are obtained by using {\bf Magma}.

\begin{example}
For $q=7$, $m=2$, $e=t=2$ and $a=6$, we have an $[8,4,2]$ cyclic code over $\F_7$ with $N=4$. The weight hierarchy of this code is as follows
$$
d_1=2, d_2=4, d_3=6, d_4=8,
$$
which coincides with the result of Theorem \ref{thm-semi}.
\end{example}

\begin{example}
For $q=7$, $m=2$, $e=t=2$ and $a=2$, we have a $[24,4,6]$ cyclic code over $\F_7$ with $N=4$. The weight hierarchy of this code is as follows
$$
d_1=6, d_2=12, d_3=18, d_4=24,
$$
which coincides with the result of Theorem \ref{thm-semi}.
\end{example}

\section{A Number-theoretic Approach to GHWs}\label{sec3}

Let $\cC$ be the cyclic code defined in Section~\ref{sec2}. In this section, we introduce a number-theoretic approach to the computation of GHWs of $\cC$. Firstly, we give a brief introduction to cyclic codes, group characters and Gauss periods. Secondly, we derive two general expressions closely related to the determination of GHWs, which will be used in our computation.

\subsection{Cyclic Codes}

Let $\cC$ be an $[n,k]$ linear code over $\Fq$ with $\gcd(n,q)=1$. $\cC$ is called a cyclic code, if $(c_0,c_1,\ldots,c_{n-1}) \in \cC$ implies its cyclic shift $(c_{n-1},c_0,\ldots,c_{n-2}) \in \cC$. For a cyclic code $\cC$, each codeword $(c_0,\ldots,c_{n-1})$ can be associated with a polynomial $\sum_{i=0}^{n-1} c_ix^i$ in the principal ideal ring $R_n:=\Fq[x]/(x^{n}-1)$. Under this correspondence, $\cC$ can be identified with an ideal of $R_n$. Hence, there is a unique monic polynomial $g(x) \in \Fq[x]$ with $g(x) \mid x^n-1$ such that $\cC=(g(x))R_n$ and $g(x)$ has the smallest degree among the elements in $\cC$. This $g(x)$ is called the {\em generator polynomial} of $\cC$, and $h(x)=\frac{x^n-1}{g(x)}$ is called the {\em parity-check polynomial} of $\cC$. When $R_n$ is specified, a cyclic code is uniquely determined by either the generator polynomial or the parity-check polynomial. $\cC$ is said to have $i$ {\em nonzeroes} if its parity-check polynomial can be factorized into a product of $i$ irreducible polynomials over $\Fq$. Thus, the cyclic codes defined in Section~\ref{sec2} may have arbitrary number of nonzeroes. A cyclic code is said to be {\em irreducible}, if it has only one nonzero. Otherwise, it is called a {\em reducible} cyclic code.

\subsection{Group Characters and Gauss Periods}

Let $q=p^s$ where $p$ is a prime number. The {\it canonical additive character} $\psi_q$ of $\Fq$ is given by
\begin{align*}
  \psi_q: \Fq & \longrightarrow \bb{C} \\
         x  & \mapsto \zp^{\Tqp(x)},
\end{align*}
where $\zp=\exp\left(2 \pi \sqrt{-1} /p\right)$ is a primitive $p$-th root of unity of $\bb{C}$, and $\Tqp$ is the trace function from $\Fq$ to $\Fp$. If $Q$ is a power of $q$, by the transitivity of trace functions, we have $\psi_Q=\psi_q \circ \TQq$.

Let $\ga$ be a primitive element of $\FQ$. For $N \mid (Q-1)$, denote by $\langle \ga^N \rangle$ the multiplicative subgroup of $\FQ^*$ generated by $\ga^N$. Then for any $0 \le i \le N-1$, $\CN_i=\ga^i \langle \ga^N \rangle$ is called the $i$-th {\it cyclotomy class} of $\FQ$. For any $a \in \FQ$, we define the {\it Gauss period} $\etaN_a$ as
$$
\etaN_a=\sum_{x \in \CN_0}\psi_Q(ax).
$$

\subsection{The First Expression}

Now, we are going to derive the first expression related to the GHWs of $\cC$.

By Delsarte's Theorem \cite{Del}, codewords of $\cC$ can be represented uniquely by $c(\ux)=(c_i(\ux))_{i=1}^n$, where $\ux=(x_1,x_2,\ldots,x_t)$ runs over the set $\FQ^t$ and
$$
c_i(\ux)=\TQq\left(\sum_{j=1}^t x_j\ga^{a_ji}\right), \quad 1 \le i \le n.
$$
In other words, the map
\begin{align*}
  \Psi: \FQ^t & \longrightarrow \cC \\
         \ux  & \mapsto c(\ux)
\end{align*}
is an isomorphism between two $\Fq$-vector spaces $\FQ^t$ and $\cC$, hence induces a 1-1 correspondence between $r$-dimensional $\Fq$-subspaces of $\FQ^t$ and $r$-dimensional subcodes of $\cC$ for any $1 \le r \le tm$. For any $\Fq$-vector space $M$, denote by $\left[{M \atop r}\right]$ the set of $r$-dimensional $\Fq$-subspaces of $M$.

For any $H_r \in \subsp$, define
$$
N(H_r)=|\{i : 1 \le i \le n, c_i(\ub)=0, \forall \ub \in H_r\}|,
$$
and for any $1 \le r \le tm$, define
$$
N_r=\max\left\{N(H_r) \mid H_r \in \subsp\right\}.
$$
Since $\Psi$ is an isomorphism, by definition, the $r$-th GHW of $\cC$ can be expressed as
\begin{equation}\label{exp-GHW0}
d_r:=d_r(\cC)=n-N_r
\end{equation}

Define $\be_j=\ga^{\frac{Q-1}{e}\De_j}$ for $1 \le j \le t$ and $g=\ga^a$. According to \cite[Section IV]{XLG}, we have the following expression
\begin{equation}\label{exp-GHW1}
N(H_r)=\frac{N}{e\de q^r} \sum_{\ub \in H_r}\sum_{h=1}^e \etaNg,
\end{equation}
where $\ub=(b_1,b_2,\ldots,b_t) \in H_r$. From now on, we always consider the case where $e=t \ge1$. In this case, the above expression can be further simplified as follows.

Define a linear transformation $\psi$ from $\FQ^t$ to $\FQ^t$ as
\begin{equation*}
\psi:\begin{pmatrix}
     b_1 \\
     b_2 \\
     \vdots \\
     b_t
  \end{pmatrix}
  \rightarrow
  \begin{pmatrix}
     y_1 \\
     y_2 \\
     \vdots \\
     y_t
  \end{pmatrix}
  =
  \begin{pmatrix}
    g & & & \\
      & g^2 &  &\\
      & & \ddots & \\
      & & & g^t
  \end{pmatrix}
  \begin{pmatrix}
    \be_1 & \be_2 & \cdots & \be_t \\
    \be_1^2  & \be_2^2 & \cdots & \be_t^2 \\
    \vdots  & \vdots & \ddots & \vdots \\
    \be_1^t & \be_2^t & \cdots & \be_t^t
  \end{pmatrix}
  \begin{pmatrix}
     b_1 \\
     b_2 \\
     \vdots \\
     b_t
  \end{pmatrix}
\end{equation*}
where
$$
y_h=g^h\sum_{j=1}^t b_j\be_j^h, 1 \le h \le t.
$$
Indeed, $\psi$ induces an isomorphism from $\FQ^t$ to $\FQ^t$
\begin{equation*}
\begin{pmatrix}
     b_1 \\
     b_2 \\
     \vdots \\
     b_t
  \end{pmatrix}
  \rightarrow
  \begin{pmatrix}
     y_1 \\
     y_2 \\
     \vdots \\
     y_t
  \end{pmatrix}
\end{equation*}
and permutes all $r$-dimensional subspaces of $\FQ^t$. Therefore, when $e=t \ge 1$, by (\ref{exp-GHW1}), we have
\begin{align*}
N(H_r)&=\frac{N}{t\de q^r} \sum_{\ub \in H_r}\sum_{h=1}^t \etaNg \\
      &=\frac{N}{t\de q^r} \sum_{\uy \in \psi(H_r)}\sum_{h=1}^t \etaNy,
\end{align*}
where $\uy=(y_1,y_2,\ldots,y_t) \in \psi(H_r)$. Note that
$$
\max\{N(H_r) \mid H_r \in \subsp\}=\max\{N(\psi^{-1}(H_r)) \mid H_r \in \subsp\}.
$$
For the sake of convenience, we rewrite $N(H_r)$ as
\begin{equation}\label{exp-GHW2}
N(H_r)=\frac{N}{t\de q^r} \sum_{\uy \in H_r}\sum_{h=1}^t \etaNy,
\end{equation}
which makes no essential difference in the computation of GHWs. This is our first expression concerning $N(H_r)$.

\subsection{The Second Expression}

In this subsection, we derive an alternative expression of $N(H_r)$ when $e=t \ge 1$. The main tool is the following bilinear form.

Let $\langle \cdot, \cdot \rangle: \FQ^t \times \FQ^t \to \Fq$ be a non-degenerate bilinear form given by
$$
\langle \ux ,\uy\rangle=\TQq\left(\sum_{i=1}^t x_iy_i\right),\quad \forall \, \ux=(x_1,\ldots,x_t), \uy=(y_1,\ldots,y_t) \in \FQ^t.
$$
Then for any $\Fq$-subspace $H$ of $\FQ^t$, we define
$$
\Hp=\{y \in \FQ^t \mid \lan \ux, \uy \ran=0, \forall \ux \in H \}.
$$
We have the following lemma.

\begin{lemma}{\rm \cite[Lemma 7]{XLG}}
Suppose $H \in \subsp$. Then $\dim_{\Fq}\Hp =tm-r$ and
$$
\frac{1}{q^{tm-r}}\sum_{\ux \in \Hp} \psi_q \left(\lan \ux, \uy\ran\right)=\begin{cases}
1 & \text{if $\uy \in H$},\\
0 & \text{if $\uy \not\in H$}.
\end{cases}
$$
\end{lemma}

By (\ref{exp-GHW2}) and the above lemma, we have the second expression concerning $N(H_r)$:
\begin{align}
N(H_r)=&\frac{N}{t\de q^r}\sum_{\uy \in H_r}\sum_{h=1}^t \etaNy\notag \\
      =&\frac{N}{t\de q^r}\sum_{\uy \in \FQ^t}\sum_{h=1}^t \etaNy \frac{1}{q^{tm-r}} \sum_{\ux \in \Hrp} \psi_q (\lan \ux, \uy \ran)\notag \\
      =&\frac{N}{t\de Q^t}\sum_{\ux \in \Hrp}\sum_{h=1}^t\sum_{z \in \CN_0}\sum_{\uy \in \FQ^t}\psi_Q(zy_h) \psi_q \left(\TQq\left(\sum_{i=1}^{t} x_iy_i\right)\right)\notag\\
      =&\frac{N}{t\de Q^t}\sum_{\ux \in \Hrp}\sum_{h=1}^t\sum_{z \in \CN_0}\sum_{y_1,\ldots,y_t \in \FQ} \psi_Q \left(zy_h+\sum_{i=1}^{t} x_iy_i\right)\notag\\
      =&\frac{N}{t\de}\sum_{\ux \in \Hrp}\sum_{h=1}^t\sum_{\substack{z \in \CN_0 \\ z+x_h=0 \\x_i=0, \forall \, i \ne h}}1\label{exp-GHW3}
\end{align}

Below, we will use this expression to determine the weight hierarchy of $\cC$.

\section{Proof of Theorem~\ref{thm-semi}}\label{sec4}

Now we are going to prove Theorem~\ref{thm-semi}. Throughout this section, we have the following assumptions.
\begin{itemize}
\item $m$ is even, $e=t \ge 1$ and $2 < N \le q^{\hm}$.
\item $j$ is the smallest positive integer such that $p^j \equiv -1 \pmod{N}$ and $\frac{sm}{2j}$ is odd.
\end{itemize}
Since $N \le q^{\hm}$ and $N \mid q^{\hm}+1$, we have $N \le \frac{q^{\hm}+1}{2}$.

For any $H_r \in \subsp$ and $1\le h \le t$, define
$$
W_h:=W_h(N)=\underbrace{\zs \times \cdots \times \zs}_{h-1} \times (-\underset{h}{\CN_0}) \times \underset{h+1}{\zs} \times \cdots \times \underset{t}{\zs}
$$
and
$$
U_h:=U_h(H_r)=\Hrp \bigcap \biggl( \underbrace{\zs \times \cdots \times \zs}_{h-1} \times \underset{h}{\FQ} \times \underset{h+1}{\zs} \times \cdots \times \underset{t}{\zs} \biggr).
$$
By (\ref{exp-GHW3}), we have
\begin{equation}\label{exp-GHW4}
N(H_r)=\frac{N}{t\de}\sum_{\ux \in \Hrp}\sum_{h=1}^t\sum_{\substack{z \in \CN_0 \\ z+x_h=0 \\x_i=0, \forall \, i \ne h}}1=\frac{N}{t\de}\sum_{h=1}^t |\Hrp \cap W_h|=\frac{N}{t\de}\sum_{h=1}^t |U_h \cap W_h|
\end{equation}

Note that $U_h$ is an $\Fq$-vector space. Given the dimension of $U_h$, as a first step, we need to consider the maximal size of the intersection $U_h \cap W_h$. To this end, the following lemma determines the maximal size of the intersection between a cyclotomy class $\CN_i$ and an $\Fq$-subspace of $\FQ$. Given a subset $A \subset \FQ$, we use $A^*$ to denote the set $A \sm \{0\}$.

\begin{lemma}\label{lem-intsemi}
  Let $j$ be the smallest positive integer such that $p^j\equiv -1 \pmod{N}$. Let $\frac{sm}{2j}$ be odd and $m$ be even. Let $0 \le l \le m$ and $L \subset \FQ$ be an $l$-dimensional $\Fq$-subspace. Define a function
$$
f(l):=\begin{cases}
  q^l-1 & \text{if $0 \le l \le \frac{m}{2}$}, \\
  \frac{q^l-1}{N}+\frac{(N-1)(q^{\frac{m}{2}}-q^{l-\frac{m}{2}})}{N} & \text{if $\frac{m}{2} \le l \le m$}.
\end{cases}
$$
Then, for $0 \le l \le m$,
$$
\max\left\{\left|L \cap \CN_i \right| : L \subset \FQ, \dimq(L)=l \right\}=f(l),
$$
where $0 \le i \le N-1$.
Furthermore, the subspace $L \subset \FQ$ that achieves the maximal value can be chosen as follows:
\begin{itemize}
\item[1)] If $0 \le l \le \hm$, then we can choose any $l$-dimensional subspace $L \subset \ga^{i}\Fqhm$;
\item[2)] If $\hm \le l \le m$, then we can choose any $l$-dimensional subspace $L$, which is a disjoint union of $q^{l-\hm}$ cosets of $\ga^i\Fqhm$, such that $\ga^i\Fqhm \subset L$.
\end{itemize}
\end{lemma}
\begin{proof}
Suppose there exists an $l$-dimensional subspace $L$ such that $|L \cap \CN_0|$ is maximal. Then, for $0 \le i \le N-1$, $\ga^{i}L$ is an $l$-dimensional subspace such that $|\ga^{i}L \cap \CN_i|=|L \cap \CN_0|$ is also maximal. Hence, it suffices to consider the case $i=0$.

Since $p^j \equiv -1 \pmod{N}$ and $\frac{sm}{2j}$ is odd, we have $q^{\hm} \equiv -1 \pmod{N}$, which implies $\Fqhm^* \subset \CN_0$. When $1 \le l \le \hm$, by choosing a subspace $L \subset \Fqhm$, we have $|L \cap \CN_0|=q^l-1$, which is clearly maximal.

When $\hm \le l \le m$, for $0 \le h \le q^{\hm}$, define $L_h=\ga^h\Fqhm$. Since $\dimq(L)=l \ge \hm$ and $\dimq(L_h)=\hm$, we have
\begin{equation}\label{ineqn-intsemi1}
|L \cap L_h^*| \ge q^{l-\hm}-1, \quad \forall \, 0 \le h \le q^{\hm}.
\end{equation}
Note that
$$
\CN_0=\cup_{v=0}^{\frac{q^{\hm}+1}{N}-1}L_{vN}^*, \quad \FQ^* \sm \CN_0=\cup_{u=1}^{N-1}\cup_{v=0}^{\frac{q^{\hm}+1}{N}-1}L_{u+vN}^*.
$$
By (\ref{ineqn-intsemi1}), we have
\begin{align*}
|L \cap (\FQ^* \sm \CN_0)|&=|L\cap(\cup_{u=1}^{N-1}\cup_{v=0}^{\frac{q^{\hm}+1}{N}-1}L_{u+vN}^*)| \\
                          &=\sum_{u=1}^{N-1}\sum_{v=0}^{\frac{q^{\hm}+1}{N}-1} |L \cap L_{u+vN}^*| \\
                          &\ge \frac{(N-1)(q^{\hm}+1)(q^{l-\hm}-1)}{N},
\end{align*}
which implies
\begin{align}
|L \cap \CN_0| &\le q^l-1-\frac{(N-1)(q^{\hm}+1)(q^{l-\hm}-1)}{N} \notag \\
               &=\frac{q^l-1}{N}+\frac{(N-1)(q^{\frac{m}{2}}-q^{l-\frac{m}{2}})}{N} \label{exp-upper}
\end{align}
Next, we are going to show that by choosing a proper subspace $L$, the upper bound (\ref{exp-upper}) can be achieved.

Let $L$ be an $l$-dimensional subspace, which consists of $q^{l-\hm}$ disjoint cosets of $\Fqhm$ and contains $\Fqhm$. Thus, we can write $L=\cup_{u=0}^{q^{l-\hm}-1} (w_u+\Fqhm)$, where $w_0=0$. For $0 \le u \le q^{l-\hm}-1$ and $1 \le h \le q^{\hm}$, we claim $|(w_u+\Fqhm) \cap L_h| \le 1$. Suppose $w_u+f_1, w_u+f_2 \in L_h$, where $f_1, f_2 \in \Fqhm$. For $1 \le h \le q^{\hm}$, since $L_h$ is a linear space, we have $(w_u+f_1)-(w_u+f_2)=f_1-f_2 \in L_h \cap \Fqhm =\{0\}$. Thus, the claim is true. Morevoer, for $1 \le h \le q^{\hm}$, we have
\begin{align*}
|L \cap L_h^*|&=|(\cup_{u=0}^{q^{l-\hm}-1} (w_u+\Fqhm)) \cap L_h^*| \\
              &=|(\cup_{u=1}^{q^{l-\hm}-1} (w_u+\Fqhm)) \cap L_h^*| \\
              &=\sum_{u=1}^{q^{l-\hm}-1} |(w_u+\Fqhm) \cap L_h^*| \\
              &\le q^{l-\hm}-1.
\end{align*}
Comparing with (\ref{ineqn-intsemi1}), we have $|L \cap L_h^*|=q^{l-\hm}-1$ for $1 \le h \le q^{\hm}$. Together with $|L \cap L_0^*|=q^{\hm}-1$, we get
\begin{align*}
|L \cap \CN_0|=&|L \cap L_0^*|+\sum_{v=1}^{\frac{q^{\hm}+1}{N}-1}|L \cap L_{vN}^*| \\
              =&\frac{q^l-1}{N}+\frac{(N-1)(q^{\frac{m}{2}}-q^{l-\frac{m}{2}})}{N}.
\end{align*}
Therefore, $|L \cap \CN_0|$ achieves the upper bound (\ref{exp-upper}).
\end{proof}

For $1 \le i \le t$, define $u_i=\dim_{\Fq}(U_i)$, where $0 \le u_i \le m$. By definition, we have
$$
H_r^{\perp} \supset \bigoplus_{i=1}^t U_i,
$$
which implies $\sum_{i=1}^{t} u_i \le tm-r$. By Lemma~\ref{lem-intsemi}, it is easy to see that the maximal size of the intersection $U_i \cap W_i$ equals $f(u_i)$. Therefore, by (\ref{exp-GHW4}), we have $N(H_r) \le \frac{N}{t\de}\sum_{i=1}^t f(u_i)$. To make $N(H_r)$ as large as possible, we must have $\sum_{i=1}^{t} u_i = tm-r$, which implies
\begin{equation}
H_r^{\perp}=\bigoplus_{i=1}^t U_i. \label{exp-subsp}
\end{equation}
Moreover, for each $1 \le i \le t$, $U_i$ is chosen so that $|U_i \cap W_i|=f(u_i)$. Consequently, we have
\begin{equation}
\max\{ N(H_r) \mid H_r \in \subsp \}=\frac{N}{t\de}\max\{ \sum_{i=1}^t f(u_i) \mid \sum_{i=1}^{t} u_i = tm-r \} \label{exp-max}
\end{equation}
From the viewpoint of (\ref{exp-GHW4}), (\ref{exp-subsp}) and (\ref{exp-max}), the subspace $H_r^{\perp}$ corresponding to the maximal $N(H_r)$ can be characterized by the sequence $\uu=(u_1, u_2,\ldots,u_t)$, where $\sum_{i=1}^{t} u_i = tm-r$. Without loss of generality, we assume that $m \ge u_1 \ge u_2 \ge \cdots \ge u_t \ge 0$. Since $1 \le r \le tm$, we write $\sum_{i=1}^{t} u_i = tm-r =r_1m+r_2$ for some unique $0 \le r_1 <t$ and $0 \le r_2 <m$.

Next, we are going to study which sequence $\uu=(u_1, u_2,\ldots,u_t)$ leads to the maximal $N(H_r)$. As a preparation, we define two operations on the sequence $\uu$. Suppose for some $1 \le i < j \le t$, $\uu=(u_1,u_2,\ldots,u_t)$ satisfies $m > u_i \ge u_j>0$. Then define an operation $\cS_{ij}$ on $\uu$ as
$$
\cS_{ij}(\uu)=(u_1,\ldots,u_{i-1},u_i+1,u_{i+1},\ldots,u_{j-1},u_j-1,u_{j+1},\ldots,u_t).
$$
Suppose $\uu$ is of the form
$$
\uu=(u_1,\ldots,u_i,\underbrace{\hm,\ldots,\hm}_{l},u_j,\ldots,u_t),
$$
where $l \ge 2$. Then define an operation $\cS$ on $\uu$ as
$$
\cS(\uu)=(m,u_1,\ldots,u_i,\underbrace{\hm,\ldots,\hm}_{l-2},u_j,\ldots,u_t,0).
$$
Furthermore, for $\uu=(u_1,u_2,\ldots,u_t)$, we define
$$
T(\uu)=\sum_{i=1}^{t} f(u_i),
$$
where $f$ is the function defined in Lemma~\ref{lem-intsemi}. Now, we have the following lemma concerning the change of the summation $T$ when the operations $\cS_{ij}$ and $\cS$ applied. Recall that $2 < N \le \frac{q^{\hm}+1}{2}$, and we use $0 \le v \le \hm-1$ to denote the unique integer such that $q^v \le \frac{q^{\hm}+1}{N}-1 <q^{v+1}$.

\begin{lemma}\label{lem-semiop}
Let $\uu=(u_1,u_2,\ldots,u_t)$ satisfy $m \ge u_1 \ge \ldots \ge u_t \ge 0$. We have the following.
\begin{itemize}
\item[1)] If $\hm \ge u_i+1>u_i\ge u_j>u_j-1 \ge 0$, then $T(\cS_{ij}(\uu))\ge T(\uu)$,
\item[2)] If $m \ge u_i+1>u_i\ge\hm\ge u_j>u_j-1 \ge 0$, then
          $$
          T(\cS_{ij}(\uu))\begin{cases}
              \ge T(\uu) & \mbox{if $u_i-u_j \ge \hm-v-1$,}\\
              \le T(\uu) & \mbox{if $u_i-u_j \le \hm-v-2$,}
          \end{cases}
          $$
\item[3)] If $m \ge u_i+1>u_i\ge u_j>u_j-1 \ge \hm$, then $T(\cS_{ij}(\uu))\ge T(\uu)$.
\item[4)] If $\uu=(u_1,\ldots,u_i,\underbrace{\hm,\ldots,\hm}_{l},u_j,\ldots,u_t)$ with $l \ge 2$, then $T(\cS(\uu)) \ge T(\uu)$.
\end{itemize}
\end{lemma}
\begin{proof}
The proof is elementary and omitted here.
\end{proof}

For the sake of convenience, we define the operation $\cS_{ij}$ in 1) of above lemma as $\cS_{ij}^{1}$. Similarly, we can define $\cS_{ij}^{2}$ and $\cS_{ij}^3$. The above lemma indicates that when the operations $\cS_{ij}^1$, $\cS_{ij}^3$ and $\cS$ are employed, the summation $T$ is nondecreasing. When the operation $\cS_{ij}^2$ is employed, the situation is more involved. To be more precise, we define an inverse operation $\wti{\cS_{ij}^2}$ of $\cS_{ij}^2$ as follows. Suppose for some $1 \le i < j \le t$, $\uu=(u_1,u_2,\ldots,u_t)$ satisfies $m \ge u_i>u_i-1 \ge \hm \ge u_j+1 > u_j \ge 0$. Then define an operation $\wti{\cS_{ij}^2}$ on $\uu$ as
$$
\wti{\cS_{ij}^2}(\uu)=(u_1,\ldots,u_{i-1},u_i-1,u_{i+1},\ldots,u_{j-1},u_j+1,u_{j+1},\ldots,u_t).
$$

Given a sequence $\uu=(u_1,u_2,\ldots,u_t)$ satisfying $m \ge u_i+1>u_i \ge \hm \ge u_j > u_j-1 \ge 0$, we have
$$
\wti{\cS_{ij}^2}(\cS_{ij}^2(\uu))=\uu.
$$
Given a sequence $\uu=(u_1,u_2,\ldots,u_t)$ satisfying $m \ge u_i>u_i-1 \ge \hm \ge u_j+1 > u_j \ge 0$, we have
$$
\cS_{ij}^2(\wti{\cS_{ij}^2}(\uu))=\uu.
$$
Hence, the operation $\wti{\cS_{ij}^2}$ can be viewed as an inverse of $\cS_{ij}^2$. The following remark restates 2) of Lemma~\ref{lem-semiop}.

\begin{remark}\label{rem-semiop}
Let $\uu=(u_1,u_2,\ldots,u_t)$ be a sequence. If $m \ge u_i+1>u_i\ge\hm\ge u_j>u_j-1 \ge 0$ and $u_i-u_j \ge \hm-v-1$, then
$$
T(\cS_{ij}^2(\uu)) \ge T(\uu).
$$
If $m \ge u_i>u_i-1 \ge \hm \ge u_j+1 > u_j \ge 0$ and $u_i-u_j \le \hm-v-2$, then
$$
T(\wti{\cS_{ij}^2}(\uu)) \ge T(\uu).
$$
\end{remark}
Therefore, if $u_i-u_j \ge \hm-v-1$ (resp. $u_i-u_j \le \hm-v-2$), the summation $T$ is nondecreasing when the operation $\cS_{ij}^2$ (resp. $\wti{\cS_{ij}^2}$) applied.

We are going to show that any sequence $\uu=(u_1,u_2,\ldots,u_t)$ can be transformed to one of a few sequences with special forms, by using operations $\cS_{ij}^1$, $\cS_{ij}^2$, $\wti{\cS_{ij}^2}$, $\cS_{ij}^3$ and $\cS$. Moreover, with the help of Lemma~\ref{lem-semiop} and Remark~\ref{rem-semiop}, we make sure that each operation used in the transformation keeps the summation $T$ nondecreasing. Hence, the sequence producing the maximal value of $N(H_r)$ is among a few sequences with special forms. Next, we describe this transformation process.

Given any sequence $\uu=(u_1,u_2,\ldots,u_t)$, for the entries greater than (resp. less than) $\hm$, we apply the operation $\cS_{ij}^1$ (resp. $\cS_{ij}^3$) repeatedly. Then, we can always get a sequence of the form
$$
\ul{u_1}=(\underbrace{m,\ldots,m}_{l_1},a,\underbrace{\hm,\ldots,\hm}_{l_2},b,\underbrace{0,\ldots,0}_{l_3}),
$$
where $m >a \ge \hm > b \ge 0$ and $l_1,l_2,l_3 \ge 0$. Applying the operation $\cS$ to $\ul{u_1}$ repeatedly, we have the one of the following two cases:
\begin{equation}
  \mbox{Case A):} \quad \mbox{$l_2$ even,} \quad \ul{u_2}=(\underbrace{m,\ldots,m}_{l_1+\frac{l_2}{2}},a,b,\underbrace{0,\ldots,0}_{l_3}), \label{exp-seq1}
\end{equation}
where $m >a \ge \hm > b \ge 0$ and
$$
  \mbox{Case B):} \quad \mbox{$l_2$ odd,} \quad \ul{u_2}=(\underbrace{m,\ldots,m}_{l_1+\frac{l_2-1}{2}},a,\hm,b,\underbrace{0,\ldots,0}_{l_3}),
$$
where $m >a \ge \hm > b \ge 0$.

Next, by using the operations $\cS_{ij}^2$, $\wti{\cS_{ij}^2}$ and $\cS$, the sequence in Case A) or Case B) can be further transformed to one of the following four cases. We have Cases A1) and A2) which can be derived from Case A) and have Cases B1) and B2) which can be derived from Case B). Recall that $tm-r=r_1m+r_2$, where $0 \le r_1 \le t-1$ and $0 \le r_2 \le m-1$. We observe that each operation involved in the transformation keeps the summation $T$ nondecreasing.
\begin{align*}
&\mbox{Case A1): $l_2$ even, $\hm \le a+b <m$, $\hm \le r_2 < m$, $r_1=l_1+\frac{l_2}{2}$} \\
&\ul{u_3}=\begin{cases}
             (\underbrace{m,\ldots,m}_{r_1},r_2,0,\ldots,0) & \mbox{if $a-b\ge\hm-v-1$} \\
             (\underbrace{m,\ldots,m}_{r_1},\hm,r_2-\hm,0,\ldots,0) & \mbox{if $a-b\le\hm-v-2$}
          \end{cases}   \\
&\mbox{Case A2): $l_2$ even, $m \le a+b <\frac{3m}{2}$, $0 \le r_2 < \hm$, $r_1=l_1+\frac{l_2}{2}+1$} \\
&\ul{u_3}=\begin{cases}
             (\underbrace{m,\ldots,m}_{r_1-1},m,r_2,0,\ldots,0) & \mbox{if $a-b\ge\hm-v-1$} \\
             (\underbrace{m,\ldots,m}_{r_1-1},r_2+\hm,\hm,0,\ldots,0) & \mbox{if $a-b\le\hm-v-2$}
          \end{cases} \\
&\mbox{Case B1): $l_2$ odd, $\hm \le a+b <m$, $0 \le r_2 < \hm$, $r_1=l_1+\frac{l_2+1}{2}$} \\
&\ul{u_3}=\begin{cases}
             (\underbrace{m,\ldots,m}_{r_1-1},r_2+\hm,\hm,0,\ldots,0) & \mbox{if $a-b\ge\hm-v-1$} \\
             (\underbrace{m,\ldots,m}_{r_1-1},m,r_2,0,\ldots,0) & \mbox{if $a-b\le\hm-v-2$}
          \end{cases} \\
&\mbox{Case B2): $l_2$ odd, $m \le a+b <\frac{3m}{2}$, $\hm \le r_2 < m$, $r_1=l_1+\frac{l_2+1}{2}$} \\
&\ul{u_3}=\begin{cases}
             (\underbrace{m,\ldots,m}_{r_1},\hm,r_2-\hm,0,\ldots,0) & \mbox{if $a-b\ge\hm-v-1$} \\
             (\underbrace{m,\ldots,m}_{r_1},r_2,0,\ldots,0) & \mbox{if $a-b\le\hm-v-2$}
          \end{cases}
\end{align*}

For instance, let us see how the sequences in Case A1) can be derived. In Case A), the sequence is of the form (\ref{exp-seq1}), where $\hm \le a+b <\frac{3m}{2}$. If $\hm \le a+b <m$, we have $r_2=a+b$, $\hm \le r_2 < m$ and $r_1=l_1+\frac{l_2}{2}$. If $a-b \ge \hm-v-1$, applying the operation $\cS_{ij}^2$ repeatedly gives
$$
\ul{u_3}=(\underbrace{m,\ldots,m}_{r_1},r_2,0,\ldots,0).
$$
If $a-b \le \hm-v-2$, applying the operation $\wti{\cS_{ij}^2}$ repeatedly gives
$$
\ul{u_3}=(\underbrace{m,\ldots,m}_{r_1},\hm,r_2-\hm,0,\ldots,0).
$$
According to Lemma~\ref{lem-semiop} and Remark~\ref{rem-semiop}, each operation involved in the transformation keeps the summation $T$ nondecreasing. Hence, the sequences in Case A1) have been obtained. Similarly, we can derive the corresponding sequences for the remaining three Cases A2), B1) and B2).

%

Therefore, we have shown that any sequence $\uu$ can be transformed to one of the above four Cases A1), A2), B1) and B2). Since each operation involved in the transformation keeps the summation $T$ nondecreasing, the sequence leading to the maximal value of $N(H_r)$ must belong to one of the four cases.

If $0 \le r_2 < \hm$, considering Cases A2) and B1), a direct computation shows
$$
T((\underbrace{m,\ldots,m}_{r_1-1},m,r_2,0,\ldots,0)) \ge T((\underbrace{m,\ldots,m}_{r_1-1},r_2+\hm,\hm,0,\ldots,0)).
$$
Thus, when $0 \le r_2 < \hm$, the sequence
$$
(\underbrace{m,\ldots,m}_{r_1},r_2,0,\ldots,0)
$$
leads to the maximal value of $N(H_r)$. Consequently, by (\ref{exp-max}), we have
\begin{align*}
N_r&=\max\{ N(H_r) \mid H_r \in \subsp\} \\
&=\frac{N}{t\de}T((\underbrace{m,\ldots,m}_{r_1},r_2,0,\ldots,0))\\
&=\frac{N}{t\de}(r_1\frac{q^m-1}{N}+f(r_2))=\frac{r_1(q^m-1)+N(q^{r_2}-1)}{t\de}.
\end{align*}

If $\hm \le r_2 < m$, considering Cases A1) and B2), a direct computation shows
$$
T((\underbrace{m,\ldots,m}_{r_1},r_2,0,\ldots,0)) \ge T((\underbrace{m,\ldots,m}_{r_1},\hm,r_2-\hm,0,\ldots,0)).
$$
Thus, when $\hm \le r_2 < m$, the sequence
$$
(\underbrace{m,\ldots,m}_{r_1},r_2,0,\ldots,0)
$$
leads to the maximal value of $N(H_r)$. Consequently, by (\ref{exp-max}), we have
\begin{align*}
N_r&=\max\{ N(H_r) \mid H_r \in \subsp \}\\
&=\frac{N}{t\de}T((\underbrace{m,\ldots,m}_{r_1},r_2,0,\ldots,0))\\
&=\frac{N}{t\de}(r_1\frac{q^m-1}{N}+f(r_2)) \\ &=\frac{r_1(q^m-1)+q^{r_2}-1+(N-1)(q^{\hm}-q^{r_2-\hm})}{t\de}.
\end{align*}

Together with (\ref{exp-GHW0}), the proof of Theorem~\ref{thm-semi} is complete.

\section{Conclusion}\label{sec5}

The generalized Hamming weights are fundamental parameters of linear codes. They convey the structural information of a linear code and determine its performance in various applications. However, the computation of the GHWs of linear codes is difficult in general. This paper is a sequel of \cite{XLG} and studies the GHWs of a family of cyclic codes introduced in \cite{YXDL}, which may have arbitrary number of nonzeroes. We determine the weight hierarchy by generalizing a number-theoretic approach proposed in \cite{YLFL}. It is worthy to note that our main theorem can be regarded as an extension of the known results concerning the weight hierarchy of semiprimitive codes.

A very interesting question is, whether the techniques in this paper can be applied to some more complicated cases. Recall that two crucial conditions in our main theorem are $p$ being semiprimitive modulo $N$ and $e=t \ge 1$. We ask if the weight hierarchy can also be computed, when $p$ is semiprimitive mod $N$ and $e > t \ge 1$, or, when $p$ modulo $N$ belongs to the Index $2$ case, namely, $p$ generates an index $2$ subgroup of the multiplicative group of units in $\Z_{N}$.

\section*{Acknowledgement}

The author wishes to express his gratitude to Professor Cunsheng Ding and Professor Maosheng Xiong, for their guidance and encouragement during his stay at the Hong Kong University of Science and Technology.



\end{document}